%% file: main.tex
\RequirePackage{tikz}  %needs to be loaded before epl2 class
\usetikzlibrary{calc}
\usetikzlibrary{decorations.markings,arrows.meta}

\documentclass[doublecol, figures, final]{epl2}
\usepackage{amsmath, amssymb, mathtools, amsthm}
\usepackage{tabularx,array,booktabs,multirow}

\usepackage{xcolor}
\definecolor{Green}{HTML}{008000}
\definecolor{Blue}{HTML}{2323FF}
\usepackage[colorlinks=True]{hyperref}
\hypersetup{allcolors=black}
\newtheoremstyle{plain}
{}{}{\itshape}{}{\bfseries}{.}{.5em}{}
\theoremstyle{plain}    % kursive Schrift
\newtheorem{theorem}{Theorem}
\newtheorem*{theorem*}{Main theorem}  % only unnumbered theorem  
\newtheorem{lemma}[theorem]{Lemma}

\newtheorem*{definition*}{Definition}
\newtheoremstyle{cr}
{}{}{\itshape}{}{\bfseries}{.}{.5em}{{\thmname{#1} (CR\thmnumber{#2})}}
\theoremstyle{cr}    % for the computational results
\newtheorem{compres}{Computational result}
\newcommand{\cA}{\mathcal{A}}
\newcommand{\cB}{\mathcal{B}}
\newcommand{\cC}{\mathcal{C}}

\newcommand{\cM}{\mathcal{M}}

\DeclareMathOperator{\Tr}{Tr}
\newcommand{\bN}{\mathbb{N}}

\tikzset{
	mid arrow/.style = {
		postaction = {
			decorate,
			decoration = {
				markings,
				mark = at position 0.5 with {\arrow[line width=0.8pt]{Stealth}}
			}
		}
	},
	c/.style = {
		fill, insert path={circle[radius=0.05]}
	}
}

\newcommand{\drawgraph}[4]{
	\def\r{1.5cm} % radius
	\def\d{0.5cm} % offset for label
	\def\x{4} % scope scaling x
	\def\y{3.97} % scope scaling y
	\begin{scope}[shift={(#2*\x,#3*\y)}]
		\node at (135:\r+\d) {#1};
		\foreach \i in {0,...,15} {
			\coordinate (\i) at ({90 - 360/16 * \i}:\r);
		}
		\foreach \i in {0,2,...,14} {
			\pgfmathtruncatemacro{\j}{\i+1}
			\draw[red, thick] (\i) -- (\j);
		}
		\foreach \i in {1,3,...,13} {
			\pgfmathtruncatemacro{\j}{\i+1}
			\draw[green!60!black, thick] (\i) -- (\j);
		}
		\draw[green!60!black, thick] (15) -- (0);
		\foreach \edge in {#4} {
			\draw[blue, thick] \edge;
		}
		\foreach \i in {0,...,15} {
			\draw[fill=black] (\i) circle (0.05);
		}
	\end{scope}
}

\newcommand{\drawgraphtwo}[4]{
	\def\r{0.75cm} % radius
	\def\d{0.5cm} % offset for label
	\def\x{4} % scope scaling x
	\def\y{3.97} % scope scaling y
	\begin{scope}[shift={(#2*\x,#3*\y)}]
		\node at (135:\r+\d) {#1};
		\foreach \i in {0,...,7} {
			\coordinate (\i) at ({90 - 360/8 * \i}:\r);
		}
		\foreach \i in {8,...,15} {
			\coordinate (\i) at ($({90 + 360/8 + 360/8 * \i}:\r) + (3*\r,0)$);
		}
		\foreach \i in {0,2,...,14} {
			\pgfmathtruncatemacro{\j}{\i+1}
			\draw[red, thick] (\i) -- (\j);
		}
		\foreach \i in {1,3,5} {
			\pgfmathtruncatemacro{\j}{\i+1}
			\draw[green!60!black, thick] (\i) -- (\j);
		}
		\draw[green!60!black, thick] (7) -- (0);
		\foreach \i in {9,11,13} {
			\pgfmathtruncatemacro{\j}{\i+1}
			\draw[green!60!black, thick] (\i) -- (\j);
		}
		\draw[green!60!black, thick] (15) -- (8);
		\foreach \edge in {#4} {
			\draw[blue, thick] \edge;
		}
		\foreach \i in {0,...,15} {
			\draw[fill=black] (\i) circle (0.05);
		}
	\end{scope}
}

\title{Low order maximally single-trace graphs as the first counter\-examples to large \boldmath$N$ factorization in random tensors}
\shorttitle{Counterexamples to large $N$ factorization in random tensors}

\author{J. Bethold \and H. Keppler}
\shortauthor{J. Bethold and H. Keppler}

\institute{                    
   Heidelberg University, Institut f{\"u}r Theoretische Physik, Philosophenweg 16, 69120 Heidelberg, Germany, EU.
  }

\abstract{
We give the first and lowest order examples of 3-regular 3-edge-colored graphs that demonstrate the non-factorization of tensor model invariants in the large $N$ limit of Gaussian random tensors. The existence of such graphs has recently been proven on general grounds, but without providing explicit examples. Moreover, it was expected that such examples require graphs with a very large number of vertices. We show that the smallest such graphs have 16 vertices.
The non-factorization for Gaussian random tensors is in stark contrast to the well-known large $N$ factorization for random matrices.
}

\begin{document}

\maketitle
% epl2 class redefines \[
% Reinstate definition:
\let\[\undefined
\DeclareRobustCommand\[{\begin{equation*}}

\section{Introduction}

Random tensors are a generalization of random matrices and have multiple applications in physics and mathematics (see refs. in \cite{Gurau:2025evo}, and \cite{Ambjorn,Sasakura:1990fs,gurau,Carrozza:2024gnh,tanasabook,GurauAIHPD,Klebanov:2018fzb} for reviews and textbook accounts). We understand random tensors as $D$-dimensional ($D\geq 3$) arrays of random variables $T_{a_1\dots a_D}$, $a_i\in\{1,\dots,N\}$, with appropriate transformation properties under basis changes, and no symmetry under permutation of the indices. Similar to random matrices ($D=2$), the expectation values of invariant observables of $O(N)^{\otimes D}$-invariant ($U(N)^{\otimes D}$ in the complex case) random tensor models  admit a $1/N$ expansion. While for matrices, this expansion is dominated by planar graphs, for tensors the expectation values of invariant observables are governed at leading order by so-called \textit{melonic} graphs \cite{Bonzom:2011zz,Gurau:2011aq,Gurau-N,Bonzom:2012hw,Gurau:2013cbh,Gurau:2013pca,Gurau:2011kk,Benedetti:2017qxl,Carrozza:2015adg,Carrozza:2018ewt,Carrozza:2021qos,Klebanov}.

For a large class of random $N\times N$ matrices, the moments (that is, expectation values of matrix invariants) factorize into products of cumulants (connected expectations) in the large $N$ limit. For example, for a Hermitian random matrix $H$ and $q_1,\dots,q_m\in \bN$, $m\geq1$,
\[
\left\langle \Tr (H^{q_1})\dots \Tr (H^{q_m})\right\rangle \xrightarrow{N\to\infty} \left\langle \Tr (H^{q_1})\right\rangle_c\dots \left\langle\Tr (H^{q_m})\right\rangle_c ,
\]
where we denote the cumulants by $\langle\, \cdot\, \rangle_c$. In other words, the joint cumulants
\[
\left\langle \Tr (H^{q_1})\dots \Tr (H^{q_m})\right\rangle_c , \quad q_1,\dots,q_m\in \bN,\; m>1,
\] are suppressed in $1/N$, relative to the product of the individual cumulants
\[
\left\langle \Tr (H^{q_1})\right\rangle_c\dots \left\langle \Tr (H^{q_m})\right\rangle_c .
\]
In physics, the large $N$ factorization is important for understanding large $N$ conformal field theories, the AdS/CFT conjecture \cite{Maldacena:1997re,Gubser:1998bc,Witten:1998qj}, and, more recently, the factorization problem in the context of Jackiw--Teitelboim gravity \cite{Maldacena:2004rf,Saad:2019lba,Blommaert:2021fob,Mukhametzhanov:2021hdi}. In the large $N$ limit, quantum fluctuations are suppressed and the theories become classical. In mathematics, it is crucial for understanding the theory of free probability as the leading large $N$ limit of moments of random matrices \cite{Voiculescu:1991,Voiculescu:1992,Nica_Speicher:2006,Bonnin:2024lha,Collins:2024pip,Nechita:2025}.

Gurau, Joos and Sudakov \cite{Gurau:2025evo} showed that the large $N$ factorization generically does not hold for Gaussian random tensors. Not all expectation values of products of tensor invariants (to be discussed in more details below) factorize into a product of cumulants in the large $N$ limit. Tensor invariants that are polynomial in the tensor entries are called trace invariants. These are in correspondence with regular edge-colored graphs $G$ where the colored edges encode tensor contractions of indices whose position is specified by the coloring. We therefore denote trace invariants by $\Tr_G(T)$. Throughout the paper, we denote by $n$ half the number of vertices of $G$, and we note that $\Tr_G(T)$ is a homogeneous polynomial of degree $2n$. The authors of \cite{Gurau:2025evo} proved that factorization does not hold by showing that if $n$ is large enough, there exist connected edge-colored graphs $G$, such that
\[ 
\left\langle\Tr_G(T)\Tr_G(T)\right\rangle
\]
does not factorize, \textit{i.e.},{$
\left\langle\Tr_G(T) \Tr_G(T)\right\rangle_c $ is larger than or equal to $ \left\langle\Tr_G(T)\right\rangle_c\left\langle\Tr_G(T)\right\rangle_c
$
in scaling with $N$.
Their proof uses techniques from probabilistic combinatorics, and they show that for very large $n$, almost all graphs $G$ are such that $\left\langle\Tr_G(T)\Tr_G(T)\right\rangle$ does not factorize. However, they do not give an explicit counterexample (as it is not relevant for the proof). Using the estimates given in Theorem 1 from \cite{Gurau:2025evo}, one naively expects counterexamples for huge graphs with $n\sim 3\cdot 10^4$. Since the bounds in \cite{Gurau:2025evo} are not tight, this number should be understood as a large overestimate.

In this paper, we give the first explicit and lowest order counterexamples to large $N$ factorization for real Gaussian random tensors with $D=3$. We find 42 such graphs already at $n=8$, \textit{i.e.}, with $16$ vertices. Moreover, we show that these are indeed the smallest graphs that provide the desired counterexamples and that there are no such graphs for $n=9$. 

The paper is structured as follows: First, we give details and definitions of random tensors, their invariants, and edge-colored graphs. Second, we detail our results, which are divided into two computational results---since we checked certain classes of graphs using a computer program---and a main theorem which together show the claimed results. Third, we prove our main theorem, and conclude and discuss our results at the end.

\section{Preliminaries}
We consider real tensors transforming in the outer tensor product of three fundamental representations of the $O(N)$ group\cite{Carrozza:2015adg}. In components, each index transforms individually:
\[
\begin{gathered}
T'_{a_1 a_2 a_3}
=
\sum_{b_1,b_2,b_3}
O^{(1)}_{a_1 b_1}
O^{(2)}_{a_2 b_2}
O^{(3)}_{a_3 b_3}
\, T_{b_1 b_2 b_3} \;,
\\
O^{(1)}, O^{(2)}, O^{(3)} \in O(N) \, .
\end{gathered}
\]
We call the position of an index its color.

\subsection{Invariants and graphs}
Invariants under the $O(N)^{\otimes 3}$ action that are polynomial in the tensor entries can be written as a linear combination of so-called \emph{trace invariants}. These are associated with 3-regular 3-edge-colored graphs $G$ that we define below for $D$ colors.
\begin{definition*}
A \emph{$D$-regular $D$-edge-colored graph} is a tuple
\[
G=(V(G),E_1(G),\dots,E_D(G)) \,,
\]
where $V(G)=\{v_1,\dots,v_{2n}\}$ is a finite vertex set, and, for $i\in\{1,\dots,D\}$, the set $E_i(G)$ is a set of edges of color $i$, such that:
\begin{itemize}
    \item $E_i(G)$ is a pairing of $V(G)$, that is, it is a partition of $V(G)$ into unordered pairs $\{v_i,v_j\}$.
    \item all vertices are $D$-valent and incident to exactly one edge of each color.
\end{itemize}
A $D$-colored graph is connected, if it is connected as an uncolored graph with edge set $E_1(G) \cup \dots \cup E_D(G)$.
\end{definition*}
In the following, we refer to these graphs as $D$-colored graphs or colored graphs for short. We will also omit the argument $G$, \textit{e.g.}, in $E_i(G)$, if it is clear to which graph we refer. 
Two $D$-colored graphs are isomorphic $G\cong G'$ if there exists a bijection (relabeling) between their vertex sets that preserves the adjacency relations and respects the coloring. 
We think of edge-colored graphs primarily in terms of their edge sets and sometimes suppress the vertex set from the notation. The union of two colored graphs $G=(V(G),E_1(G),\dots,E_D(G))$, and $G'=(V(G'),E_1(G'),\dots,E_D(G'))$ is defined as the union of their edge and vertex sets, respecting the colors:
\[
\begin{aligned}
    G\cup G' \coloneqq (&V(G)\cup V(G'),E_1(G)\cup E_1(G'),\\
    &\ \dots,E_D(G)\cup E_D(G')) \;.
\end{aligned}
\]
We write $G\sqcup G'$ if we want to emphasize that the vertex and edge sets of $G$ and $G'$ are disjoint. For a $D$-colored graph $G$, and a pairing of its vertices $M$, we define the $(D+1)$-colored graph $G\cup M$, by adding the pairs in $M$ as edges of a new color:
\[
G\cup M \coloneqq (V(G),E_1(G),\dots,E_D(G),M)
\]
For $i\neq j$, $i,j\in\{1,\dots,D\}$, we call the connected components of the 2-colored subgraph $G_{ij}=(V,E_i,E_j)$ \textit{faces} of colors $(i,j)$. Equivalently, a face of colors $(i,j)$ is a cycle in $G$ whose edges alternate between colors $i$ and $j$. We denote the number of faces of colors $(i,j)$ in $G$ by $F_{ij}(G)$. 
A $D$-colored graph $G$ is said to be \textit{maximally single-trace} (MST)\footnote{The term was introduced in the tensor model literature in \cite{Ferrari:2017jgw}. In the mathematical literature, such colorings are referred to as perfect one-factorizations \cite{Wallis:1997}.} iff
\[
F_{ij}(G)=1\;,\quad \forall i\neq j,\ i,j\in\{1,\dots,D\}\;.
\]

From now on $D=3$. To each $3$-colored graph $G$ we associate a trace invariant $\Tr_G(T)$ and vice versa. More precisely, each vertex represents a tensor and for each color $i$ the edge set $E_i$ encodes the contraction of the tensor indices of the same color:
\begin{equation}\label{eq:traceinv}
\Tr_{G}(T)\coloneqq
\kern-.5em
\sum_{a_1^1,\dots, a_{3}^{2n}}\kern-.5em \bigg(
\prod_{p=1}^{2n} T_{a_1^{p} a_2^{p} a_3^{p}}
\bigg)
\prod_{i=1}^{3}
\bigg(
\prod_{\{k,l\}\in E_i} \delta_{a_i^{k} a_i^{l}}
\bigg)\,.
\end{equation}

\subsection{Gaussian random tensors and large $N$ scaling}
A random tensor model is an invariant probability measure on the tensor entries. A Gaussian measure is specified by its covariance
\begin{equation}\label{eq:covariance}
\langle T_{a_1 a_2 a_3}\, T_{b_1 b_2 b_3} \rangle
= \frac{1}{N^\nu}
\delta_{a_1 b_1}\delta_{a_2 b_2}\delta_{a_3 b_3}\;,
\end{equation}
and, using Wick's theorem, higher moments are computed in terms of the covariance and a sum over pairings\footnote{Sometimes called perfect matchings in graph theory and in \cite{Gurau:2025evo}.} (see, \textit{e.g.}, \cite{gurau}). The factor $N^{-\nu}$ is purely conventional and our results do not depend on it. A common choice is $\nu=D-1$.

We can represent the pairings $M$ in Wick's theorem by edges of a new color $0$. The combination $G\cup M$ can be viewed as a $(3+1)$-colored graph. Throughout the paper, we denote the total number of faces containing the color $0$ by
\[
F(M,G) = \sum_{i=1}^{3} F_{0i}(G\cup M)\;.
\]
This number determines the scaling in $N$ of the Gaussian expectation value of a trace invariant. 
In computing such an expectation value for a fixed Wick pairing $M$, Kronecker deltas following the edges of colors $i\in\{1,2,3\}$ in \eqref{eq:traceinv} combine with those from the covariance \eqref{eq:covariance} that follow the color $0$. Together, they form a free sum, hence a factor of $N$, for each cycle whose edges alternate between colors $0$ and $i$. These are exactly the faces of colors $(0,i)$ in $G\cup M$. This leads to the expression for the Gaussian expectation value of a trace invariant (see, \textit{e.g.}, \cite{Gurau:2025evo,gurau}): 
\[
\langle \Tr_{G}(T) \rangle
=
\frac{1}{N^{\nu n}}
\sum_{M\in\cM_n}
N^{F(M,G)}\;,
\]
where the sum runs over the set of all pairings $\cM_n$ on $2n$ vertices, carrying the new color $0$. 

\subsection{Large $N$ factorization}

Let $\{G_1,\dots,G_q\}$ be a family of connected $3$-colored graphs. We say that the expectation of their products factorizes in the large $N$ limit if
\[
\left\langle \Tr_{G_1}(T)\dots \Tr_{G_q}(T) \right\rangle
= \prod_{\rho=1}^{q} \left\langle \Tr_{G_\rho}(T) \right\rangle \left(1 + O(N^{-1})\right)\;.
\]
Note that for $G_\rho$ connected, $\left\langle \Tr_{G_\rho}(T) \right\rangle\equiv\left\langle \Tr_{G_\rho}(T) \right\rangle_c$. 
It was shown in \cite{Gurau:2025evo}, Lemma 1, that the expectation values of products of trace invariants always factorize in the large $N$ limit if and only if every 3-colored graph $G$ satisfies
\begin{equation}\label{eq:maxfaces}
\max_{M \in \cM_n}
F(M,G) > \frac{3n}{2} \;.
\end{equation}
The authors of \cite{Gurau:2025evo} showed that very large graphs generically violate this inequality, and therefore not all such expectations factorize.

\section{Results and main theorem}

Our results are twofold. First, we tested the inequality \eqref{eq:maxfaces} for a restricted class of colored graphs up to $n=9$ using a computer program. Crucially, we found 42 graphs that violate the inequality. For each such graph,
\[ \left\langle \Tr_G(T)\Tr_G(T)\right\rangle_c \]
is equal in scaling in $N$ to 
\[ \left\langle \Tr_G(T)\right\rangle_c\left\langle\Tr_G(T)\right\rangle_c ,\]
hence the expectation does not factorize. Secondly, our main theorem complements this survey by asserting that all colored graphs with $n\leq9$ that have not been tested always satisfy \eqref{eq:maxfaces}. It follows that these 42 graphs are the smallest and only graphs that violate \eqref{eq:maxfaces} up to $n=9$.
We now give these results in full detail:

\begin{figure}[t]
	\begin{tikzpicture}[font=\itshape]
	\drawgraph{}{0}{0}{(0)--(2)(1)--(5)(3)--(7)(4)--(11)(6)--(8)(9)--(13)(10)--(15)(12)--(14)}
	\end{tikzpicture}
	\centering
	\caption{A 3-regular 3-edge-colored MST graph on $2n=16$ vertices with $\max_M F =12$.}
	\label{fig:graphex}
\end{figure}
\begin{figure}[h!]
	\begin{tikzpicture}[font=\itshape]
				\drawgraphtwo{}{0}{0}{(0)--(2)(1)--(8)(3)--(10)(4)--(6)(9)--(11)(13)--(15)}
				\draw[blue, thick] (7) to[out=45,in=135] (14);
				\draw[blue, thick] (5) to[out=-45,in=-135] (12);
	\end{tikzpicture}
	\centering
	\caption{The exceptional 3-regular 3-edge-colored non MST graph on $2n=16$ vertices with $\max_M F =12$. Notable are the four triangle subgraphs with one edge of each color.}
	\label{fig:graphexx}
\end{figure}
\begin{table}[h!]
    \caption{The number of 3-edge-colored 3-regular MST graphs, $\max_M F(M,G)$ ($\max F$), and the number of 3-edge-colored 3-regular connected graphs (this is OEIS A002831 \cite{OEIS}).}
    \label{tab:res}
    \vspace{0.5ex}
    \renewcommand{\arraystretch}{1.1}
    \begin{tabularx}{\columnwidth}{>{\centering\arraybackslash}p{1cm} >{\raggedleft\arraybackslash}p{1.5cm} >{\raggedleft\arraybackslash}p{1.5cm} X }
    \toprule
    $n$ & A002831 & MST & $\max F$  \\
    \midrule
    1 & 1  & 1    & 3  \\
    2 & 4  & 1    & 4  \\
    3 & 11 & 2    & 6  \\
    4 & 60 & 4    & 7  \\
    \midrule
    \multirow{2}*{5} & \multirow{2}*{318} & \multirow{2}*{24}
             & 8 \quad \small (2 times) \\
      &  &   & 9 \quad \small(22 times) \\
    \midrule
    6 & 2,806 & 125  & 10 \\
    \midrule
    \multirow{2}*{7} & \multirow{2}*{29,359} & \multirow{2}*{1,161}
             & 11 \quad \small(279 times) \\
      &  &   & 12 \quad \small(882 times) \\
    \midrule
    \multirow{2}*{8} & \multirow{2}*{396,196} & \multirow{2}*{12,504} 
             & \textbf{12} \quad \small\textbf{(41 times)} \\
      &  &   & 13 \quad \small(12,463 times) \\
    \midrule
    \multirow{2}*{9} & \multirow{2}*{6,231,794} & \multirow{2}*{167,782} 
             & 14 \quad \small(62,475 times) \\
      &  &   & 15 \quad \small(105,307 times) \\
    \bottomrule
    \end{tabularx}
    \ \\[-.5ex]
    \textit{\small Note: The MST graph for $n=2$ is known as tetrahedron, and the ones at $n=3$ as wheel and prism in the tensor model literature. MST graphs have been counted with \cite{keppler_cgraphc_2026}.}
\end{table}

\begin{compres}\label{cr:1}
 For $n\leq 9$ we checked that all 3-regular 3-edge-colored graphs $G$ with one face for one pair of colors and an arbitrary number for the remaining two pairs have $\max_M F(M,G)>\frac{3n}{2}$, except for $n=8$ where there exists 41 graphs with $\max_M F(M,G)=\frac{3n}{2}=12$. These graphs are all MST, and non-bipartite. As an example, we show one in fig.~\emph{\ref{fig:graphex}} and all of them in the appendix.

The searched for sector of graphs with one face for one pair of colors contains all MST graphs on $2n$ vertices ($n\leq 9$). We furthermore counted these graphs and computed $\max_M F(M,G)$ for each graph. These results are shown in table~\emph{\ref{tab:res}}.
\end{compres}
This result was obtained by explicitly checking all \mbox{$(2n-1)!!$} pairings $M$ of color $0$ for all MST graphs with $n\leq9$, using a computer program written in C++. The source code used in this work is publicly available at heiDATA \cite{keppler_cgraphc_2026}. Below, we describe the idea of the underlying algorithm:
\begin{enumerate}
    \item Fix the integer $n$. The set of vertices\footnote{In C++ the labeling starts at $0$.} is $V=\{1,2,\dots,2n\}$.
    \item Recursively compute all pairings on $2n$ vertices by using a standard backtracking algorithm. This gives the set $\cM_n$ of cardinality $(2n-1)!!$.
    \item Fix the set of color~$1$ edges to be $E_1=\{\{1,2\},\{3,4\},$ $\dots,\{2n-1,2n\}\}$, and the set of color~$2$ edges $E_2=\{\{2,3\},\{4,5\},\dots,\{2n,1\}\}$. Choosing the set of color $3$ edges $E_3$ as an element from $\cM_n$, defines a connected 3-regular 3-colored graph $G=(V,E_1,E_2,E_3)$ with $F_{12}=1$.
    \item Loop through all $(2n-1)!!$ choices for $E_3\in\cM_n$. This ensures that the loop covers all edge-colored graphs $G$ with the properties stated before, but several such graphs $G$ will be isomorphic to one another.

    Inside the main loop (over graphs $G$):
    \begin{enumerate}
        \item Use the nauty library\footnote{Nauty is a set of procedures for computing automorphism groups of graphs and canonical graph labelling (including edge-colored graphs) that can be used to check whether two graphs are isomorphic. See \cite{MCKAY201494} and the documentation \cite{nautydoc}. We used version 2.9.3.
        } to check whether the current graph $G$ is isomorphic to a graph $G'$ that was already encountered in the loop. If $G\cong G'$, continue to the next graph; if not, continue with the next step below.
        \item Compute $F_{13}$ and $F_{23}$ by counting the cycles of the graphs with edges $(E_1,E_2)$ and $(E_2,E_3)$. If both are equal to one, the graph is MST.
        \item Loop through all choices of color $0$ pairings $M\in\cM_n$ for fixed $G$. For each $M$, compute\footnote{We actually precompute $F_{01}$ and $F_{02}$ for all $M$ before the main loop. This is possible since $E_1$ and $E_2$ are fixed for all $G$.} $F(M,G)$ and determine the maximum value $\max_{M}F(M,G)$.
        \item If $\max_{M}F(M,G)\leq \frac{3n}{2}$, the graph, the information whether it is MST, $\max_{M}F(M,G)$, and the maximizer are saved. 
    \end{enumerate}
\end{enumerate}
Notice that the arbitrary vertex labeling and color assignments are irrelevant.

\begin{compres}\label{cr:2}
	Among all 3-regular 3-edge-colored graphs $G$ with $n=8$, exactly two faces for each pair of colors, all faces of length eight, and every two faces sharing exactly two edges if they have exactly one color in common, there exists a single graph with $\max_M F(M,G)=\frac{3n}{2}=12$. This graph is shown in fig.~\emph{\ref{fig:graphexx}}.
\end{compres}
This result was obtained by a slight variation of the program described above, starting with $E_1$ and $E_2$ forming two faces of colors $(1,2)$---each of length eight.

\begin{theorem*}
    For all connected 3-edge-colored graphs $G$ on $2n$ vertices with $n\leq9$ that are neither MST nor the graph in fig.~\emph{\ref{fig:graphexx}}, there exists a pairing $M$ of color $0$ such that \mbox{$F>3n/2$}, i.e., the inequality \eqref{eq:maxfaces} holds.
\end{theorem*}
\begin{proof}
    This follows from Lemmas~\ref{lem:13}--\ref{lem:89} and our computational results.
    The case $n\leq 5$ is covered by Lemma~\ref{lem:13} and \ref{lem:35}, and $6\leq n\leq 7$ is covered by Lemma~\ref{lem:67}. 
    Lemma~\ref{lem:89} shows the claim for all graphs with $8\leq n\leq 9$, at least two faces for any of the three pairs of colors, and excluding the graph in fig.~\ref{fig:graphexx}. It partially utilizes CR\ref{cr:2}.

    It remains to consider graphs with $8\leq n\leq 9$, a single face for one pair of colors and two faces for the remaining pairs of colors, or two faces for one pair of colors and a single face for the remaining pairs. 
    This case is covered by CR\ref{cr:1}, where we checked that all such graphs with a single face for one pair of colors admit a color $0$ pairing $M$, such that $F(M,G)>\frac{3n}{2}$.
\end{proof}
The proof uses several Lemmas that we detail in the next section. 

\section{Proof of the main theorem}
To show that \eqref{eq:maxfaces} holds for a particular graph $G$, one has to construct a pairing $M$ of color $0$ edges, such that $G\cup M$ has sufficiently many faces of colors $(0,1)$, $(0,2)$, and $(0,3)$.
In this section, we discuss how this can be obtained for increasingly higher $n$.
    
Choosing the color $0$ edges of $M$ to be parallel to the color $1$ edges of $G$ gives $n$ faces of colors $(0,1)$ and, since there is at least one $(1,2)$ face and one $(1,3)$ face, this gives at least $n+2$ faces containing the color $0$. We thus have:
\begin{lemma}\label{lem:13}
    For all $3$-edge-colored graphs $G$ on $2n$ vertices, with $n\leq3$, there exists a pairing $M$ of color $0$ edges, such that $F(M,G)=n+2>3n/2$.
\end{lemma}
\begin{proof}
    As just discussed, setting $M=E_1(G)$ one obtains $F(M,G)=n+2$, and $n+2>3n/2$ for $n\leq3$.
\end{proof}
By making assumptions on the structure of $G$, in particular its number of faces $F_{ij}$, we can trivially increase the lower bound on $\max_M F(M,G)$.
\begin{lemma}\label{lem:35}
    For all $3$-edge-colored graphs $G$ on $2n$ vertices, with $3\leq n\leq5$ that are not MST, there exists a pairing $M$ of color $0$ edges, such that $F(M,G)=n+3>3n/2$.
\end{lemma}
\begin{proof}
     Without loss of generality let $F_{12}\geq2$, $F_{13}\geq1$, and $F_{23}\geq1$. Choose $M$ such that the color $0$ edges are parallel to the color $1$ edges. This gives $n$ faces of colors $(0,1)$. As by assumption there are at least two faces $(1,2)$, one face of colors $(1,3)$, and the color $0$ edges are parallel to the color $1$ edges we get at least $n+3$ faces. This is strictly greater than $3n/2$ for $n\leq 5$.
     All remaining graphs have $F_{12}=F_{13}=F_{23}=1$ and are thus MST.
\end{proof}
For higher $n$, $F(M,G)$ must increase further, since the right hand side of \eqref{eq:maxfaces} grows as $\frac{3n}{2}$. To continue our analysis, it is therefore necessary to adapt $M$ to the structure of $G$. The following Lemma describes how our naive choice of $M$ can be slightly altered to increase $F$ by one.

\begin{figure}
    \centering
    \begin{tikzpicture}[font=\footnotesize]
    	\coordinate (a) at (0,0);
    	\coordinate (b) at (0,-1.5);
    	\coordinate (c) at (0,-3);
    	\coordinate (d) at (0,-4.5);
    	\coordinate (e) at ($(a)+(160:1.5)$);
    	\coordinate (f) at ($(b)+(-1.5,0)$);
    	\coordinate (g) at ($(c)+(-1.5,0)$);
    	\coordinate (h) at ($(d)+(-160:1.5)$);
    	
    	\draw[mid arrow] (e) -- (a) node[midway,above] {3} [c];
    	\draw[mid arrow] (b) [c] -- (f) node[midway,below] {3};
    	\draw[mid arrow] (g) -- (c) node[midway,above] {3} [c];
    	\draw[mid arrow] (d) [c] -- (h) node[midway,below] {3};
    	\draw[mid arrow] (a) -- (b) node[midway, right] {1} node[pos=0.5,left] {\normalsize$e$};
    	\draw[mid arrow] (c) -- (d) node[midway, right] {1} node[pos=0.5,left] {\normalsize$f$};
    	\draw (a) node[below left] {\normalsize$v_1$} 
    	(b) node[above left] {\normalsize$v_2$}
    	(c) node[below left] {\normalsize$w_1$}
    	(d) node[above left] {\normalsize$w_2$};
    	\draw (a) --++ (20:1.5) node[midway,above] {2} [c]
    	(d) --++ (-20:1.5) node[midway,below] {2} [c];
    	
    	\draw [dotted, very thick]
    	($(a) + (20:1.5)$) .. controls ++(2.5,0.5) and ++(2.5,-0.5) .. ($(d) + (-20:1.5)$)
    	node[midway,left,rotate=-90,anchor=south]{(1,2)} node[pos=.25,left=5pt] {\normalsize$\mathcal{C}$}
    	(b) to[out=-30,in=30] node[rotate=-90,above] {(1,2)} (c);
    	
    	\draw [dotted, very thick, mid arrow] (h) .. controls ++(-2.5,-0.5) and ++(-2.5,0.5) .. (e) node[midway,right,rotate=90,anchor=south]{(1,3)} node[pos=.75,right=5pt] {\normalsize$\mathcal{C}'$};
    	\draw [dotted, very thick, mid arrow] (f) to [out=180, in=180] node[rotate=90,above] {(1,3)} (g);
    \end{tikzpicture}
    \caption{Illustration of the face structure considered in the proof of Lemma~\ref{lem:flip}. See there for further explanation. The dotted lines indicate how the edges are connected along the rest of the faces $\cC$ (of colors $(1,2)$) and $\cC'$ (of colors $(1,3)$).}
    \label{fig:cycles1}
\end{figure}
    
\begin{lemma}\label{lem:flip}
    Consider a face $\cC$ of colors $(i,j)$ and a second face $\cC'$ of colors $(i,k)$ with all colors $i,j,k$ distinct. Assume further that $\cC$ and $\cC'$ share at least three edges of color $i$.
    Let $M$ be a pairing of color $0$, such that all its edges are parallel to the edges of color $i$, i.e., $M=E_i$. In this situation, one can always find a pairing $M'$ that differs from $M$ by only two edges, such that $F(M',\cC\cup \cC')=F(M,\cC\cup \cC')+1$. 
\end{lemma}
\begin{proof}
    Without loss of generality let $i=1$, $j=2$, and $k=3$. As each face is a cycle of edges, there are two natural directions one can assign to its edges: clockwise or counterclockwise. We fix a direction for the $(1,3)$ face $\cC'$, such that all its edges point in the same direction along that face. This induces a direction for all color $1$ edges $e\in E_1(\cC)\cap E_1(\cC')$ that are shared between the two faces. Viewed as edges of the $(1,2)$ face $\cC$, two shared color $1$ edges can therefore either point in the same, or in opposite directions along the face $\cC$. By assumption $|E_1(\cC)\cap E_1(\cC')|\geq 3$, and among three shared color $1$ edges, at least two of these need to point in the same direction along $\cC$. For illustration see fig.~\ref{fig:cycles1}. Let $e=(v_1,v_2)$ and $f=(w_1,w_2)$ be two such edges, directed from the first to the second vertex. It should now be clear that $v_2$ follows $v_1$ and $w_2$ follows $w_1$ when going around $\cC'$ in the chosen direction, and the same is true when going around $\cC$ in the direction fixed by both of these edges (fig.~\ref{fig:cycles1}).
    
    By construction of $M$, the color $0$ faces of $(\cC\cup\cC')\cup M$ are identical to the $(1,2)$ and $(1,3)$ faces of $\cC\cup\cC'$, and $M$ contains two color $0$ edges $e_0$ and $f_0$ that are copies of $e$ and $f$. See fig.~\ref{fig:cycles2}. 
    To construct $M'$, let $e'_0=(v_1,w_2)$, $f'_0=(w_1,v_2)$ and set 
    \[ M'=\left(M-\{e_0,f_0\}\right) \sqcup \{e'_0,f'_0\} . \]
    The new edges $e'_0$ and $f'_0$ are chosen in such a way, that they split both the previous $(0,3)$ and $(0,2)$ face into two (fig.~\ref{fig:cycles2}).
    \begin{figure}
    \centering
    \begin{tikzpicture}[font=\footnotesize]
    	\coordinate (a) at (0,0);
    	\coordinate (b) at (0,-1.5);
    	\coordinate (c) at (0,-3);
    	\coordinate (d) at (0,-4.5);
    	\coordinate (e) at ($(a)+(160:1.5)$);
    	\coordinate (f) at ($(b)+(-1.5,0)$);
    	\coordinate (g) at ($(c)+(-1.5,0)$);
    	\coordinate (h) at ($(d)+(-160:1.5)$);
    	
    	%color 0 edges
    	\draw[dashed, very thick, Blue, font=\normalsize] (a) to[out=-10,in=10] node[right] {$e_0$} (b)
    	(c) to[out=-10,in=10] node[right] {$f_0$} (d);
    	\draw[dashed, very thick, Green, font=\normalsize] (a) .. controls ++(3,0) and ++(3,0) .. (d) node[midway,right] {$e'_0$}
    	(b) .. controls ++(2,0) and ++(2,0) .. (c) node[midway,right] {$f'_0$};
    	
    	\draw[mid arrow] (e) -- (a) node[midway,above] {3} [c];
    	\draw[mid arrow] (b) [c] -- (f) node[midway,below] {3};
    	\draw[mid arrow] (g) -- (c) node[midway,above] {3} [c];
    	\draw[mid arrow] (d) [c] -- (h) node[midway,below] {3};
    	\draw[mid arrow] (a) -- (b) node[midway, right] {1} node[pos=0.5,left] {\normalsize$e$};
    	\draw[mid arrow] (c) -- (d) node[midway, right] {1} node[pos=0.5,left] {\normalsize$f$};
    	\draw (a) node[below left] {\normalsize$v_1$} 
    	(b) node[above left] {\normalsize$v_2$}
    	(c) node[below left] {\normalsize$w_1$}
    	(d) node[above left] {\normalsize$w_2$};
    	\draw (a) --++ (20:1.5) node[midway,above] {2} [c]
    	(d) --++ (-20:1.5) node[midway,below] {2} [c];
    	
    	\draw [dotted, very thick]
    	($(a) + (20:1.5)$) .. controls ++(2.5,0.5) and ++(2.5,-0.5) .. ($(d) + (-20:1.5)$)
    	node[midway,left,rotate=-90,anchor=south]{(1,2)} node[pos=.25,left=5pt] {\normalsize$\mathcal{C}$}
    	(b) to[out=-30,in=30] node[rotate=-90,above] {(1,2)} (c);
    	
    	\draw [dotted, very thick, mid arrow] (h) .. controls ++(-2.5,-0.5) and ++(-2.5,0.5) .. (e) node[midway,right,rotate=90,anchor=south]{(1,3)} node[pos=.75,right=5pt] {\normalsize$\mathcal{C}'$};
    	\draw [dotted, very thick, mid arrow] (f) to [out=180, in=180] node[rotate=90,above] {(1,3)} (g);
    \end{tikzpicture}
    \caption{Illustration of the construction of the pairing $M'$ in the proof of Lemma~\ref{lem:flip}. See there for further explanations. We only show the color $0$ edges that are different in $M$ (in blue) and $M'$ (in green). All other color $0$ edges are parallel to the color $1$ edges. As before, the dotted lines indicate how the edges are connected along the faces.}
    \label{fig:cycles2}
    \end{figure}
    Compared to $(\cC\cup\cC')\cup M$, $(\cC\cup\cC')\cup M'$ has one more $(0,3)$ face, one more $(0,2)$ face an one less $(0,1)$ face. This finishes the proof.
\end{proof}
We are now in the position to prove the following results for $n$ up to $9$.
\begin{lemma}\label{lem:67}
    For all connected $3$-edge-colored graphs on $2n$ vertices with $6\leq n\leq 7$ that are not MST, there exists a pairing $M$ of color $0$ edges, such that $F=n+4>3n/2$.
\end{lemma}
\begin{proof}
    First assume that the graph has at least two faces for any pair of colors. Choose $M$ such that the color $0$ edges are parallel to the color $1$ edges. This gives $n$ faces of colors $(0,1)$. As by assumption we have $F_{12}\geq 2$, $F_{13}\geq 2$, and the color $0$ edges are parallel to the color $1$ edges we get at least $n+4$ faces. This is strictly greater than $3n/2$ for $n\leq 7$.
    
    Second, consider graphs with two faces for two pairs of colors and a single face for the remaining pair, or a single face for two pairs of colors and two faces for the remaining pair. Without loss of generality let $F_{12}=2$, $F_{13}=1$, and $F_{23}\in\{1,2\}$. If $F_{23}=2$ we can place all color $0$ edges parallel to the color $2$ edges and obtain $n+4$ faces, as before.

    It remains to consider graphs $G$ with $F_{12}=2$, and $F_{13}=F_{23}=1$. 
    Choosing again $M$, such that the color $0$ edges are parallel to the color $1$ edges we obtain $F(M,G)=n+3$.
    If one of the $(1,2)$ faces has at least three color $1$ edges, these edges are necessarily all shared with the single $(1,3)$ face, and by Lemma~\ref{lem:flip} we can obtain $M'$ such that $F(M',G)=n+4$.
    For the graphs under consideration, this is indeed always the case:    
    Denote the two $(1,2)$ faces by $\cC$ and $\cC'$. We have
    \[
    E_1(G)=E_1(\cC\sqcup\cC')=E_1(\cC)\sqcup E_1(\cC') \, ,
    \]
    as every color $1$ edge belongs either to $\cC$ or to $\cC'$. Let $|E_1(\cC)|=p$, $|E_1(\cC)|=q$, and since $|E_1(G)|=n$ we always have $\max\{p,q\}\geq \lceil n/2\rceil$ ($\geq 3$ for $n\geq 6$). 
    This concludes the proof.
\end{proof}

\begin{lemma}\label{lem:89}
    For all connected $3$-edge-colored graphs on $2n$ vertices with $8\leq n\leq 9$, at least two faces for any of the three pairs of colors, and except the graph in fig.~\emph{\ref{fig:graphexx}} there exists a pairing $M$ of color~$0$ edges, such that $F=n+5>3n/2$.
\end{lemma}
\begin{proof}
    First, without loss of generality let $F_{12}\geq 3$, $F_{13}\geq 2$, and $F_{23}\geq 2$. Choose $M$ such that the color $0$ edges are parallel to the color $1$ edges. This gives $n$ faces of colors $(0,1)$, at least three faces of colors $(0,2)$, and at least two faces of colors $(0,3)$. In total, we get at least $n+5$ faces. This is $>3n/2$ for $n\leq 9$.
    
    It remains to consider graphs with $F_{12}=F_{13}=F_{23}=2$. Choosing the color $0$ edges of $M$ to be parallel to the color $1$ edges, we obtain $F(M,G)=n+4$ by the same argument. To prove the Lemma, we thus have to gain at least one additional face.
    For each pair of colors, denote the two faces of colors $(i,j)$ by $\cC_{ij}$ and $\cC'_{ij}$. Let $i\neq j\neq k$ be the three different colors. Since $E_i(G)=E_i(\cC_{ij}\sqcup \cC'_{ij})$, we have
    \[
    |E_i(\cC_{ik})| = |E_i(\cC_{ij})\cap E_i(\cC_{ik})| + |E_i(\cC'_{ij})\cap E_i(\cC_{ik})| \, ,
    \]
    and the same for $\cC_{ik}$ replaced by $\cC'_{ik}$. Since all $n$ edges of color $i$ have to be either in $\cC_{ik}$ or $\cC'_{ik}$, we find the condition
    \[
    \begin{aligned}
    	n={} &|E_i(\cC_{ik})| + |E_i(\cC'_{ik})|\\
    	={} &|E_i(\cC_{ij})\cap E_i(\cC_{ik})| + |E_i(\cC'_{ij})\cap E_i(\cC_{ik})| \\
    	&+ |E_i(\cC_{ij})\cap E_i(\cC'_{ik})| + |E_i(\cC'_{ij})\cap E_i(\cC'_{ik})| \,.
    \end{aligned}
    \]
    Thus, for $n=8$ either all intersections on the right hand side are equal to $2$, or at least one is $\geq 3$; and for $n=9$ at least one of the intersection has to be $\geq 3$.
    
    In the case where one of the $(i,j)$ faces shares at least three color $i$ edges with one of the $(i,k)$ faces, we can invoke Lemma~\ref{lem:flip} to obtain a color $0$ pairing $M'$, such that $F(M',G)=n+5$, as required. It only remains to consider the case $n=8$ and all intersections 
    $|E_i(\cA)\cap E_i(\cB)|=2$, $\forall \cA\in \{\cC_{ij},\cC'_{ij}\}$, $\cB\in \{\cC_{ik},\cC'_{ik}\}$,
    \textit{i.e.}, all faces that have exactly one color in common share exactly two edges. In particular, this implies that all faces have length eight. CR\ref{cr:2} shows that all such graphs, except the one in fig.~\ref{fig:graphexx} satisfy $\max_M F(M,G)>12$, therefore concluding the proof.
\end{proof}

\section{Conclusion}
We identified 42 graphs $G$ with $2n=16$ vertices, for which products of their corresponding tensor model invariants $\Tr_G(T)$ do not factorize in the large $N$ limit. For these graphs, both
\[ \left\langle\Tr_G(T)\Tr_G(T)\right\rangle_c 
\quad \text{and}\quad \left\langle\Tr_G(T)\right\rangle_c\left\langle\Tr_G(T)\right\rangle_c \]
scale with $N^{24-8\nu}$, and thus the latter term does not dominate at large $N$. 41 of these graphs are MST. As an example, one is shown in fig.~\ref{fig:graphex} and we show all of them in the appendix. The single non-MST graph with the stated property is shown in fig.~\ref{fig:graphexx}.

We showed that for $n\leq 9$ all colored graphs that are neither MST nor the graph in fig.~\ref{fig:graphexx} satisfy the inequality \eqref{eq:maxfaces}. Using a computer program, we checked all 3-regular 3-edge-colored graphs with $n\leq 9$ that have a single face for one pair of colors---including all MST graphs. Together, these results imply that these 42 colored graphs are all the counterexamples to the large $N$ factorization of the expectation values of trace invariants of Gaussian 3-index $O(N)^{\otimes 3}$ invariant random tensors up to 18 vertices.

A posteriori, we see that being MST, \textit{i.e.}, having $F_{12}(G)=F_{13}(G)=F_{23}(G)=1$, leads to a relatively small $F(M,G)$ when $G$ is paired with a pairing $M$ of color $0$. However, being MST is not sufficient even for $n\geq 8$. Among the 12,504 MST graphs with $n=8$, only 41 violate \eqref{eq:maxfaces}, and no colored graph with $n=9$ does so. It would be interesting to find out what property distinguishes these 41 graphs from other MST graphs. Also, as shown by the graph in fig.~\ref{fig:graphexx}, being MST is not necessary to violate \eqref{eq:maxfaces}.
Its face structure is such that it does not allow to increase $F(M,G)$ by optimizing $M$ beyond the naive choice of placing all edges of $M$ parallel to, \textit{e.g.}, the color $1$ edges. 
This graph is non-bipartite and planar\footnote{But it is not planar in the sense of ribbon graphs, see, \textit{e.g.}, \cite{Gurau:2025evo}.} in the usual graph theoretic sense.

Our results should be viewed as an effect of small numbers. For larger $n$, it becomes increasingly difficult to increase $F(M,G)$, and to achieve a maximum, $M$ has to be tailored to the structure of the edge-colored graph $G$ which can become very complex. For large random graphs, the expectation value of $F$ grows as $(1+o(1))n$ with probability going to $1$ \cite{Gurau:2025evo}. In other words, compared to their size, large graphs rarely have many faces which, by \eqref{eq:maxfaces}, would be required for factorization. 
This observation is at the heart of the proof in \cite{Gurau:2025evo}.

Let us note that for $8\leq n\leq 9$, the proof of our main theorem partially utilizes our computational results. On the one hand, we would like a purely analytic proof, \textit{e.g.}, by using Lemma~\ref{lem:flip} twice. On the other hand, we found examples of graphs where this strategy is not possible, and one would have to distinguish many different cases. Since the non-MST graphs with $8\leq n \leq 9$ are very close to violating \eqref{eq:maxfaces}, these complications are not surprising. 
For $n=10$, we found again several graphs that violate the inequality. 

Although large $N$ factorization does not hold generically for Gaussian random tensors, certain families of invariants, \textit{e.g.}, those given by melonic graphs, do obey large $N$ factorization \cite{Gurau:2025evo}. Let us stress that, as noted in \cite{Gurau:2025evo}, to show factorization for a specific trace invariant it is not sufficient to check the inequality \eqref{eq:maxfaces}, since the proof of Lemma 1 in \cite{Gurau:2025evo} assumes \eqref{eq:maxfaces} for all graphs.
It would be desirable to find other families of invariants for which large $N$ factorization holds. 

\acknowledgments
We thank Razvan Gurau for very helpful discussions and Alicia Castro and R.~Gurau for comments on the manuscript. 
We thank the authors of the program Nauty \cite{MCKAY201494} for their work on the practical computation of graph isomorphisms. 
H.~K.~has been supported by the Deutsche Forschungsgemeinschaft (DFG, German Research Foundation) under Germany's Excellence Strategy EXC--2181/1--390900948 (the Heidelberg STRUCTURES Cluster of Excellence).
The authors acknowledge support by the state of Baden-Württemberg through bwHPC, in particular bwUniCluster 3.0.
\medskip

\bibliography{references.bib}
\onecolumn

\begin{center}
	\textbf{Appendix}
\end{center}

On the following pages, we give all 3-regular 3-edge-colored graphs on $2n=16$ vertices with $\max_M F(M,G) =12$. These are the smallest graphs that provide counterexamples to large $N$ factorization in random tensor models. There are 42 such graphs. We shown them in \mbox{figs.~\ref{fig:graphs1} \& \ref{fig:graphs2}}, and give a list of their edges in table~\ref{tab:graphs} for the graphs \textit{1--41}. 
\linebreak Graph \textit{42} is different, its edge lists are:
\begin{align*}
	E_1&=\{\{1,2\},\{3,4\},\dots,\{15,16\}\}\;,\\
	E_2&=\{\{2,3\},\{4,5\},\{6,7\},\{8,1\},\{10,11\},\{12,13\},\{14,15\},\{16,9\}\}\;, \\ E_3&=\{\{1,3\},\{2,9\},\{4,11\},\{5,7\},\{6,13\},\{8,15\},\{10,12\},\{14,16\}\}\;.
\end{align*}

\begin{figure}[p]
	\centering
	\input{graphs1.tex}
	\caption{All 3-regular 3-edge-colored graphs on $2n=16$ vertices with $\max_M F =12$.}
	\label{fig:graphs1}
\end{figure}
\begin{figure}[p]
	\centering
	\input{graphs2.tex}
	\caption{Continuation of Fig.~\ref{fig:graphs1}.}
	\label{fig:graphs2}
\end{figure}

\newpage

\renewcommand{\arraystretch}{1.19}
\setlength{\tabcolsep}{4pt}
\begin{table}[p]
	\caption{List of the graphs \textit{1--41} from figs~\ref{fig:graphs1} \& \ref{fig:graphs2}.}
	\centering
	\label{tab:graphs}
	\input{graphs_table}

	\ \\[.5ex]
	\parbox{0.72\textwidth}{\textit{Note: 
			We fix $E_1=\{\{1,2\},\{3,4\},\dots,\{2n-1,2n\}\}$, $E_2=\{\{2,3\},\{4,5\},\dots,\{2n,1\}\}$, and only give the edges of color $3$ (blue in figs.~\emph{\ref{fig:graphs1}} \& \emph{\ref{fig:graphs2}}). 
			}
	}
\end{table}
\end{document}

%% file: graphs1.tex
\begin{tikzpicture}[font=\itshape]
\drawgraph{1}{0}{0}{(0)--(2)(1)--(5)(3)--(7)(4)--(11)(6)--(8)(9)--(13)(10)--(15)(12)--(14)}
\drawgraph{2}{1}{0}{(0)--(2)(1)--(5)(3)--(7)(4)--(13)(6)--(8)(9)--(14)(10)--(12)(11)--(15)}
\drawgraph{3}{2}{0}{(0)--(2)(1)--(5)(3)--(7)(4)--(12)(6)--(8)(9)--(11)(10)--(14)(13)--(15)}
\drawgraph{4}{3}{0}{(0)--(2)(1)--(5)(3)--(7)(4)--(15)(6)--(12)(8)--(13)(9)--(11)(10)--(14)}
\drawgraph{5}{0}{-1}{(0)--(2)(1)--(5)(3)--(7)(4)--(15)(6)--(11)(8)--(10)(9)--(13)(12)--(14)}
\drawgraph{6}{1}{-1}{(0)--(2)(1)--(5)(3)--(7)(4)--(15)(6)--(10)(8)--(13)(9)--(11)(12)--(14)}
\drawgraph{7}{2}{-1}{(0)--(2)(1)--(5)(3)--(8)(4)--(6)(7)--(11)(9)--(13)(10)--(15)(12)--(14)}
\drawgraph{8}{3}{-1}{(0)--(2)(1)--(5)(3)--(8)(4)--(6)(7)--(13)(9)--(14)(10)--(12)(11)--(15)}
\drawgraph{9}{0}{-2}{(0)--(2)(1)--(5)(3)--(8)(4)--(6)(7)--(13)(9)--(11)(10)--(15)(12)--(14)}
\drawgraph{10}{1}{-2}{(0)--(2)(1)--(5)(3)--(8)(4)--(6)(7)--(11)(9)--(14)(10)--(12)(13)--(15)}
\drawgraph{11}{2}{-2}{(0)--(2)(1)--(5)(3)--(9)(4)--(15)(6)--(11)(7)--(13)(8)--(10)(12)--(14)}
\drawgraph{12}{3}{-2}{(0)--(2)(1)--(5)(3)--(9)(4)--(15)(6)--(14)(7)--(11)(8)--(13)(10)--(12)}
\drawgraph{13}{0}{-3}{(0)--(2)(1)--(5)(3)--(9)(4)--(15)(6)--(14)(7)--(12)(8)--(10)(11)--(13)}
\drawgraph{14}{1}{-3}{(0)--(2)(1)--(5)(3)--(10)(4)--(6)(7)--(15)(8)--(13)(9)--(11)(12)--(14)}
\drawgraph{15}{2}{-3}{(0)--(2)(1)--(5)(3)--(10)(4)--(15)(6)--(8)(7)--(13)(9)--(11)(12)--(14)}
\drawgraph{16}{3}{-3}{(0)--(2)(1)--(5)(3)--(11)(4)--(6)(7)--(13)(8)--(10)(9)--(15)(12)--(14)}
\drawgraph{17}{0}{-4}{(0)--(2)(1)--(5)(3)--(11)(4)--(15)(6)--(8)(7)--(13)(9)--(14)(10)--(12)}
\drawgraph{18}{1}{-4}{(0)--(2)(1)--(5)(3)--(12)(4)--(6)(7)--(15)(8)--(10)(9)--(14)(11)--(13)}
\drawgraph{19}{2}{-4}{(0)--(2)(1)--(5)(3)--(11)(4)--(15)(6)--(14)(7)--(9)(8)--(13)(10)--(12)}
\drawgraph{20}{3}{-4}{(0)--(2)(1)--(5)(3)--(13)(4)--(9)(6)--(8)(7)--(11)(10)--(15)(12)--(14)}
\drawgraph{21}{0}{-5}{(0)--(2)(1)--(5)(3)--(13)(4)--(15)(6)--(8)(7)--(12)(9)--(11)(10)--(14)}
\drawgraph{22}{1}{-5}{(0)--(2)(1)--(5)(3)--(14)(4)--(6)(7)--(9)(8)--(13)(10)--(12)(11)--(15)}
\drawgraph{23}{2}{-5}{(0)--(2)(1)--(5)(3)--(13)(4)--(15)(6)--(11)(7)--(9)(8)--(12)(10)--(14)}
\drawgraph{24}{3}{-5}{(0)--(2)(1)--(5)(3)--(14)(4)--(6)(7)--(11)(8)--(13)(9)--(15)(10)--(12)}
\end{tikzpicture}

%% file: graphs2.tex
\begin{tikzpicture}[font=\itshape]
\drawgraph{25}{0}{-6}{(0)--(2)(1)--(5)(3)--(14)(4)--(6)(7)--(12)(8)--(10)(9)--(15)(11)--(13)}
\drawgraph{26}{1}{-6}{(0)--(2)(1)--(5)(3)--(14)(4)--(10)(6)--(11)(7)--(9)(8)--(12)(13)--(15)}
\drawgraph{27}{2}{-6}{(0)--(2)(1)--(5)(3)--(14)(4)--(10)(6)--(8)(7)--(12)(9)--(11)(13)--(15)}
\drawgraph{28}{3}{-6}{(0)--(2)(1)--(6)(3)--(5)(4)--(8)(7)--(13)(9)--(11)(10)--(15)(12)--(14)}
\drawgraph{29}{0}{-7}{(0)--(2)(1)--(6)(3)--(5)(4)--(8)(7)--(11)(9)--(13)(10)--(15)(12)--(14)}
\drawgraph{30}{1}{-7}{(0)--(2)(1)--(6)(3)--(5)(4)--(11)(7)--(13)(8)--(10)(9)--(15)(12)--(14)}
\drawgraph{31}{2}{-7}{(0)--(2)(1)--(6)(3)--(5)(4)--(10)(7)--(15)(8)--(13)(9)--(11)(12)--(14)}
\drawgraph{32}{3}{-7}{(0)--(2)(1)--(6)(3)--(5)(4)--(12)(7)--(15)(8)--(13)(9)--(11)(10)--(14)}
\drawgraph{33}{0}{-8}{(0)--(2)(1)--(6)(3)--(5)(4)--(10)(7)--(12)(8)--(14)(9)--(11)(13)--(15)}
\drawgraph{34}{1}{-8}{(0)--(2)(1)--(7)(3)--(5)(4)--(11)(6)--(8)(9)--(13)(10)--(15)(12)--(14)}
\drawgraph{35}{2}{-8}{(0)--(2)(1)--(7)(3)--(8)(4)--(6)(5)--(11)(9)--(13)(10)--(15)(12)--(14)}
\drawgraph{36}{3}{-8}{(0)--(2)(1)--(7)(3)--(9)(4)--(11)(5)--(13)(6)--(15)(8)--(10)(12)--(14)}
\drawgraph{37}{0}{-9}{(0)--(2)(1)--(7)(3)--(10)(4)--(6)(5)--(15)(8)--(13)(9)--(11)(12)--(14)}
\drawgraph{38}{1}{-9}{(0)--(2)(1)--(7)(3)--(14)(4)--(9)(5)--(11)(6)--(8)(10)--(12)(13)--(15)}
\drawgraph{39}{2}{-9}{(0)--(2)(1)--(7)(3)--(14)(4)--(12)(5)--(9)(6)--(11)(8)--(10)(13)--(15)}
\drawgraph{40}{3}{-9}{(0)--(2)(1)--(9)(3)--(14)(4)--(6)(5)--(15)(7)--(11)(8)--(13)(10)--(12)}
\drawgraph{41}{0.5}{-10}{(0)--(2)(1)--(11)(3)--(7)(4)--(15)(5)--(10)(6)--(8)(9)--(13)(12)--(14)}
\drawgraphtwo{42}{2.2}{-10}{(0)--(2)(1)--(8)(3)--(10)(4)--(6)(9)--(11)(13)--(15)}
\draw[blue, thick] (7) to[out=45,in=135] (14);
\draw[blue, thick] (5) to[out=-45,in=-135] (12);
\end{tikzpicture}

%% file: graphs_table.tex
\begin{tabular}{>{\itshape}c cccccccc}
\toprule
\normalfont Graph no. & \multicolumn{8}{c}{Edges of color 3} \\
\midrule
1 & \{1, 3\}, & \{2, 6\}, & \{4, 8\}, & \{5, 12\}, & \{7, 9\}, & \{10, 14\}, & \{11, 16\}, & \{13, 15\} \\
2 & \{1, 3\}, & \{2, 6\}, & \{4, 8\}, & \{5, 14\}, & \{7, 9\}, & \{10, 15\}, & \{11, 13\}, & \{12, 16\} \\
3 & \{1, 3\}, & \{2, 6\}, & \{4, 8\}, & \{5, 13\}, & \{7, 9\}, & \{10, 12\}, & \{11, 15\}, & \{14, 16\} \\
4 & \{1, 3\}, & \{2, 6\}, & \{4, 8\}, & \{5, 16\}, & \{7, 13\}, & \{9, 14\}, & \{10, 12\}, & \{11, 15\} \\
5 & \{1, 3\}, & \{2, 6\}, & \{4, 8\}, & \{5, 16\}, & \{7, 12\}, & \{9, 11\}, & \{10, 14\}, & \{13, 15\} \\
6 & \{1, 3\}, & \{2, 6\}, & \{4, 8\}, & \{5, 16\}, & \{7, 11\}, & \{9, 14\}, & \{10, 12\}, & \{13, 15\} \\
7 & \{1, 3\}, & \{2, 6\}, & \{4, 9\}, & \{5, 7\}, & \{8, 12\}, & \{10, 14\}, & \{11, 16\}, & \{13, 15\} \\
\midrule
8 & \{1, 3\}, & \{2, 6\}, & \{4, 9\}, & \{5, 7\}, & \{8, 14\}, & \{10, 15\}, & \{11, 13\}, & \{12, 16\} \\
9 & \{1, 3\}, & \{2, 6\}, & \{4, 9\}, & \{5, 7\}, & \{8, 14\}, & \{10, 12\}, & \{11, 16\}, & \{13, 15\} \\
10 & \{1, 3\}, & \{2, 6\}, & \{4, 9\}, & \{5, 7\}, & \{8, 12\}, & \{10, 15\}, & \{11, 13\}, & \{14, 16\} \\
11 & \{1, 3\}, & \{2, 6\}, & \{4, 10\}, & \{5, 16\}, & \{7, 12\}, & \{8, 14\}, & \{9, 11\}, & \{13, 15\} \\
12 & \{1, 3\}, & \{2, 6\}, & \{4, 10\}, & \{5, 16\}, & \{7, 15\}, & \{8, 12\}, & \{9, 14\}, & \{11, 13\} \\
13 & \{1, 3\}, & \{2, 6\}, & \{4, 10\}, & \{5, 16\}, & \{7, 15\}, & \{8, 13\}, & \{9, 11\}, & \{12, 14\} \\
14 & \{1, 3\}, & \{2, 6\}, & \{4, 11\}, & \{5, 7\}, & \{8, 16\}, & \{9, 14\}, & \{10, 12\}, & \{13, 15\} \\
\midrule
15 & \{1, 3\}, & \{2, 6\}, & \{4, 11\}, & \{5, 16\}, & \{7, 9\}, & \{8, 14\}, & \{10, 12\}, & \{13, 15\} \\
16 & \{1, 3\}, & \{2, 6\}, & \{4, 12\}, & \{5, 7\}, & \{8, 14\}, & \{9, 11\}, & \{10, 16\}, & \{13, 15\} \\
17 & \{1, 3\}, & \{2, 6\}, & \{4, 12\}, & \{5, 16\}, & \{7, 9\}, & \{8, 14\}, & \{10, 15\}, & \{11, 13\} \\
18 & \{1, 3\}, & \{2, 6\}, & \{4, 13\}, & \{5, 7\}, & \{8, 16\}, & \{9, 11\}, & \{10, 15\}, & \{12, 14\} \\
19 & \{1, 3\}, & \{2, 6\}, & \{4, 12\}, & \{5, 16\}, & \{7, 15\}, & \{8, 10\}, & \{9, 14\}, & \{11, 13\} \\
20 & \{1, 3\}, & \{2, 6\}, & \{4, 14\}, & \{5, 10\}, & \{7, 9\}, & \{8, 12\}, & \{11, 16\}, & \{13, 15\} \\
21 & \{1, 3\}, & \{2, 6\}, & \{4, 14\}, & \{5, 16\}, & \{7, 9\}, & \{8, 13\}, & \{10, 12\}, & \{11, 15\} \\
\midrule
22 & \{1, 3\}, & \{2, 6\}, & \{4, 15\}, & \{5, 7\}, & \{8, 10\}, & \{9, 14\}, & \{11, 13\}, & \{12, 16\} \\
23 & \{1, 3\}, & \{2, 6\}, & \{4, 14\}, & \{5, 16\}, & \{7, 12\}, & \{8, 10\}, & \{9, 13\}, & \{11, 15\} \\
24 & \{1, 3\}, & \{2, 6\}, & \{4, 15\}, & \{5, 7\}, & \{8, 12\}, & \{9, 14\}, & \{10, 16\}, & \{11, 13\} \\
25 & \{1, 3\}, & \{2, 6\}, & \{4, 15\}, & \{5, 7\}, & \{8, 13\}, & \{9, 11\}, & \{10, 16\}, & \{12, 14\} \\
26 & \{1, 3\}, & \{2, 6\}, & \{4, 15\}, & \{5, 11\}, & \{7, 12\}, & \{8, 10\}, & \{9, 13\}, & \{14, 16\} \\
27 & \{1, 3\}, & \{2, 6\}, & \{4, 15\}, & \{5, 11\}, & \{7, 9\}, & \{8, 13\}, & \{10, 12\}, & \{14, 16\} \\
28 & \{1, 3\}, & \{2, 7\}, & \{4, 6\}, & \{5, 9\}, & \{8, 14\}, & \{10, 12\}, & \{11, 16\}, & \{13, 15\} \\
\midrule
29 & \{1, 3\}, & \{2, 7\}, & \{4, 6\}, & \{5, 9\}, & \{8, 12\}, & \{10, 14\}, & \{11, 16\}, & \{13, 15\} \\
30 & \{1, 3\}, & \{2, 7\}, & \{4, 6\}, & \{5, 12\}, & \{8, 14\}, & \{9, 11\}, & \{10, 16\}, & \{13, 15\} \\
31 & \{1, 3\}, & \{2, 7\}, & \{4, 6\}, & \{5, 11\}, & \{8, 16\}, & \{9, 14\}, & \{10, 12\}, & \{13, 15\} \\
32 & \{1, 3\}, & \{2, 7\}, & \{4, 6\}, & \{5, 13\}, & \{8, 16\}, & \{9, 14\}, & \{10, 12\}, & \{11, 15\} \\
33 & \{1, 3\}, & \{2, 7\}, & \{4, 6\}, & \{5, 11\}, & \{8, 13\}, & \{9, 15\}, & \{10, 12\}, & \{14, 16\} \\
34 & \{1, 3\}, & \{2, 8\}, & \{4, 6\}, & \{5, 12\}, & \{7, 9\}, & \{10, 14\}, & \{11, 16\}, & \{13, 15\} \\
35 & \{1, 3\}, & \{2, 8\}, & \{4, 9\}, & \{5, 7\}, & \{6, 12\}, & \{10, 14\}, & \{11, 16\}, & \{13, 15\} \\
\midrule
36 & \{1, 3\}, & \{2, 8\}, & \{4, 10\}, & \{5, 12\}, & \{6, 14\}, & \{7, 16\}, & \{9, 11\}, & \{13, 15\} \\
37 & \{1, 3\}, & \{2, 8\}, & \{4, 11\}, & \{5, 7\}, & \{6, 16\}, & \{9, 14\}, & \{10, 12\}, & \{13, 15\} \\
38 & \{1, 3\}, & \{2, 8\}, & \{4, 15\}, & \{5, 10\}, & \{6, 12\}, & \{7, 9\}, & \{11, 13\}, & \{14, 16\} \\
39 & \{1, 3\}, & \{2, 8\}, & \{4, 15\}, & \{5, 13\}, & \{6, 10\}, & \{7, 12\}, & \{9, 11\}, & \{14, 16\} \\
40 & \{1, 3\}, & \{2, 10\}, & \{4, 15\}, & \{5, 7\}, & \{6, 16\}, & \{8, 12\}, & \{9, 14\}, & \{11, 13\} \\
41 & \{1, 3\}, & \{2, 12\}, & \{4, 8\}, & \{5, 16\}, & \{6, 11\}, & \{7, 9\}, & \{10, 14\}, & \{13, 15\} \\
\bottomrule
\end{tabular}